\documentclass[11pt]{llncs}
\usepackage{a4wide}

\newcommand{\eps}{\varepsilon}

\newcommand{\C}{\mathcal C}
\newcommand{\opt}{\mbox{\textsc{opt}}}
\newcommand{\sol}{\mbox{\textsc{sol}}}

\begin{document}
\title{A unified framework for designing EPTAS's for load balancing on parallel machines\thanks{This research was supported by a grant from the GIF, the German-Israeli Foundation for Scientific Research and Development (grant number I-1366-407.6/2016).}
}
\author{Ishai Kones\inst{1} \and Asaf Levin\inst{1}}
\institute{Faculty of Industrial Engineering and Management, The Technion, 32000 Haifa,
Israel. \email{ishai.kones@gmail.com,levinas@ie.technion.ac.il}.}
\maketitle

\begin{abstract}
We consider a general load balancing problem on parallel machines.
Our machine environment in particular generalizes the standard
models of identical machines, and the model of uniformly related
machines, as well as machines with a constant number of types, and
machines with activation costs. The objective functions that we
consider contain in particular the makespan objective and the
minimization of the $\ell_p$-norm of the vector of loads of the
machines both with possibly job rejection.

We consider this general model and design an efficient polynomial
time approximation scheme (EPTAS) that applies for all its
previously studied special cases. This EPTAS improves the current
best approximation scheme for some of these cases where only a
polynomial time approximation scheme (PTAS) was known into an
EPTAS.
\end{abstract}

\section{Introduction}
We consider a model that generalizes many previously studied optimization problems in the framework of scheduling and (minimization) load balancing problems on parallel machines.  We use this generalization in order to exhibit that there is a standard way to design efficient polynomial time approximation schemes for all these special cases and for new special cases as well.  In the earlier works, approximation schemes for many of special cases of our model were developed using ad-hoc tricks, we show that such ad-hoc methods are not necessary. 

Before going into the details of the definition of our model, we define the types of approximation schemes.   A $\rho$-approximation algorithm for a minimization problem is a polynomial time algorithm that always finds a feasible solution of cost at most $\rho$ times the cost of an optimal solution. The infimum value of $\rho$ for which an algorithm is a $\rho$-approximation is called the approximation ratio or the performance guarantee of the algorithm. A polynomial time approximation scheme (PTAS) for a given problem is a family of approximation algorithms such that the family has a $(1+\eps)$-approximation algorithm for any $\eps >0$. An efficient polynomial time approximation scheme (EPTAS) is a PTAS whose time complexity is upper bounded by the form $f(\frac{1}{\eps}) \cdot poly(n)$ where $f$ is some computable (not necessarily polynomial) function and $poly(n)$ is a polynomial of the length of the (binary) encoding of the input.
A fully polynomial time approximation scheme (FPTAS) is a stricter concept. It is defined like an EPTAS, with the added restriction that $f$ must be a polynomial in $\frac 1{\eps}$. Note that whereas a PTAS may have time complexity of the form $n^{g(\frac{1}{\eps})}$, where $g$ is for example linear or even exponential, this cannot be the case for an EPTAS. The notion of an EPTAS is modern and finds its roots in the FPT (fixed parameter tractable) literature (see e.g. \cite{CT97,DF99,FG06,Marx08}).

Since these problems are proven to be strongly NP-Hard \cite{Garey:1979:CIG:578533} (as for example our model is an extension of the minimum makespan problem on identical machines), it is unlikely (impossible assuming $\textmd{P}\neq \mathnormal{NP}$) that an optimal polynomial time algorithm or an FPTAS will be found for them. In our research, we will focus on finding an EPTAS for this general model and as a bi-product, obtain improved results to many of its special cases.  As usual, in order to present an EPTAS we can show that for a sufficiently small value of $\eps$ there exists an algorithm of time complexity of the form $f(\frac{1}{\eps}) \cdot poly(n)$ with an approximation ratio of $1+\kappa \eps$ for an arbitrary constant $\kappa$ (independent of $\eps$).

\paragraph{{\bf Our model}.}  Being a scheduling problem, the definition of the problem can be partitioned into the characteristics of the machines, the properties of the jobs, and the objective function.

\paragraph{Machines characteristics.}  We are given $m$ machines  denoted as $\{ 1,2,\ldots ,m\}$ each of which can be activated to work in one of $\tau$ types denoted as $1,2,\ldots ,\tau$.   The type of the machine will influence the processing time of a job assigned to that machine.  The input defines for every machine $i$ a (positive rational)
speed $s_i$ and an activation cost function $\alpha_i(t)$ that is a non-negative rational number denoting the cost of activating machine $i$ in type $t$.  We are also given a budget $\hat{A}$ on the total activation cost of all machines.   The meaning of this budget is that a feasible solution needs to specify for every machine $i$ its type $t_i$ such that the total activation cost is at most the budget, that is, the following constraint holds $$ \sum_{i=1}^m \alpha_i(t_i) \leq \hat{A} \ . $$
In our work we assume that $\tau$ is a constant while $m$ is a part of the input. Furthermore, without loss of generality we assume that $1=s_1\geq s_2\geq \cdots s_m >0$.

\paragraph{Jobs characteristics.}  There are $n$ jobs denoted as $J=\{ 1,2,\ldots ,n\}$.  Job  $j$ is associated with a size ($\tau$-dimensional) vector $p_j$ that specify the size $p_j(t)$ of job $j$ if it is assigned to a machine of type $t$.  That is, if job $j$ is assigned to machine $i$, and we activate machine $i$ in type $t$, then the processing time of job $j$ (on this machine) is $\frac{p_j(t)}{s_i}$.  Furthermore, for every job $j$ we are given a rejection penalty $\pi_j$ that is a positive rational number denoting the cost of not assigning job $j$ to any machine.  A definition of a feasible solution specifies for every job $j$ if $j$ is rejected (and thus incurs a rejection penalty of $\pi_j$) or not and if it is not rejected (i.e., $j$ is accepted), then the machine $i$ that $j$ is assigned to.  Formally, we need to specify a {\it job assignment} function $\sigma:J \rightarrow \{ 0,1,2,\ldots ,m\}$, where $\sigma(j)=0$ means that $j$ is rejected, and $\sigma(j)=i$ for $i\geq 1$ means that $j$ is assigned to machine $i$.

\paragraph{Definition of the objective function.}  As stated above a feasible solution defines a type $t_i$ for every machine $i$, and a job assignment function $\sigma$.  The load of machine $i$ in this solution is
$$\Lambda_i=\frac{\sum_{j\in J : \sigma(j)=i} p_{j}(t_i)}{s_i} \ . $$   Our objective function is specified using a function $F$ defined over the vector of the loads of the machines $F(\Lambda_1,\Lambda_2, \ldots ,\Lambda_m)$ that is the assignment cost of the jobs to the machines. $F$ is defined by two scalar parameters $\phi > 1$ and $1 \geq \psi \geq 0$ as follows:
$$F(\Lambda_1,\Lambda_2, \ldots ,\Lambda_m) = \psi \cdot \max_{i=1}^m \Lambda_i  + (1-\psi) \cdot \sum_{i=1}^m (\Lambda_i)^{\phi}  \ .$$
The value of $\psi$ has the following meaning.  For $\psi=1$, the
value of $F$ is the makespan of the schedule, i.e., the maximum
load of any machine, while for $\psi=0$, the value of $F$ is the
sum of the $\phi$ powers of the loads of the machines, an
objective that is equivalent to the $\ell_{\phi}$ norm of the
vector of loads.  For $\psi$ that is strictly between $0$ and $1$,
the value of $F$ is a convex combination of these classical
objectives in the load balancing literature.   The common values
of $\phi$ that were motivated by various applications that were
considered in the literature are $\phi=2$ and $\phi=3$.

Our objective is to find a type $t_i$ for every machine $i$ such that  $ \sum_{i=1}^m \alpha_i(t_i) \leq \hat{A}$, and a job assignment $\sigma$ so that the following objective function (denoted as $obj$) will be minimized:
$$obj = F(\Lambda_1,\Lambda_2, \ldots ,\Lambda_m)  + \sum_{j\in J: \sigma(j)=0} \pi_j \ . $$

Our result is an EPTAS for this load balancing problem.  For ease
of notation we denote this problem by $P$ and let $\eps>0$ be such
that $1/\eps \geq 100$ is an integer.  We will use the fact that
$\phi$ is a constant and the following simple property throughout
the analysis.
\begin{lemma}\label{claim_prop_F}
Given a value of $\rho>1$ and two vectors $(\Lambda_1,\ldots
,\Lambda_m)$ and $(\Lambda'_1,\ldots ,\Lambda'_m)$ such that for
every $i$ we have $\Lambda_i \leq \Lambda'_i \leq (1+\eps)^{\rho}
\Lambda_i$, then $$F(\Lambda_1,\Lambda_2, \ldots ,\Lambda_m) \leq
F(\Lambda'_1,\Lambda'_2, \ldots ,\Lambda'_m)\leq (1+\eps)^{\rho
\cdot \phi} F(\Lambda_1,\Lambda_2, \ldots ,\Lambda_m) \ .$$
\end{lemma}

\begin{proof}
Using the definition of $F$ we have
\begin{eqnarray*}
F(\Lambda_1,\Lambda_2, \ldots ,\Lambda_m) &=& \psi \cdot \max_{i=1}^m \Lambda_i  + (1-\psi) \cdot \sum_{i=1}^m (\Lambda_i)^{\phi} \\
&\leq& \psi \cdot \max_{i=1}^m \Lambda'_i  + (1-\psi) \cdot \sum_{i=1}^m (\Lambda'_i)^{\phi} \\
&=& F(\Lambda'_1,\Lambda'_2, \ldots ,\Lambda'_m) \\
&\leq & \psi \cdot \max_{i=1}^m (1+\eps)^{\rho} \Lambda_i  + (1-\psi) \cdot \sum_{i=1}^m ((1+\eps)^{\rho}\Lambda_i)^{\phi}\\
&\leq & (1+\eps)^{\rho\cdot \phi} \cdot \left( \psi \cdot \max_{i=1}^m \Lambda_i  + (1-\psi) \cdot \sum_{i=1}^m (\Lambda_i)^{\phi} \right) \\
&=& (1+\eps)^{\rho\cdot \phi} \cdot F(\Lambda_1,\Lambda_2, \ldots ,\Lambda_m) .
\end{eqnarray*}
\qed
\end{proof}

\paragraph{Special cases of our model and related literature on these cases.}
The objective function we consider here generalizes the makespan
minimization objective (the special case with all $\pi_j=\infty$
and $\psi=1$), the sum of the $\phi$ powers of the machines loads
(the special case with all $\pi_j=\infty$ and $\psi=0$), as well
as these two objectives with job rejections (i.e., finite $\pi_j$
for some $j\in J$).

As for the machines model that we consider, next we state some of the earlier studied special cases of this model.
We say that {\em machines have pre-specified type} if $\hat{A}=0$ and for every $i$ we have a value $t_i$ such that $\alpha_i(t_i)=0$ and $\alpha_i(t)=1$ if $t\neq t_i$.  This special case of the machine environment is the case of unrelated machines with a constant number of types, whose special case where machines have a common speed was studied in \cite{JM17} who presented an EPTAS for the makespan objective (the extension of this scheme to machines of different speeds was explored in \cite{JM17v2}).  This EPTAS of  \cite{JM17} improves earlier PTAS's for that special cases \cite{BW12,WBB13,GJKS16}.  The $\ell_p$-norm minimization objective for the case where machines have pre-spcified type and all speeds are $1$ admits a PTAS \cite{BW12}.

The case where machines have pre-specified type generalizes its special case of {\em uniformly related machines} that is the case where $\tau=1$.  For this machines model, Jansen \cite{Ja10} presented an EPTAS for the makespan objective improving the earlier PTAS established in the seminal work of Hochbaum and Shmoys \cite{HS88}, while Epstein and Levin \cite{EL13} presented an EPTAS for the minimization of the $\ell_p$-norm of the vector of machines loads improving an earlier PTAS by Epstein and Sgall \cite{ES04}.  Later on, Epstein and Levin \cite{DBLP:journals/corr/EpsteinL14} presented an EPTAS for another scheduling objective, namely total weighted completion times, and their scheme for the case where all jobs are released at time $0$, implies a different EPTAS for the minimization of the sum of squares of the loads of the machines on uniformly related machines.  As far as we know the two schemes of \cite{EL13,DBLP:journals/corr/EpsteinL14} are the only examples for EPTAS's for load balancing objectives on uniformly related machines where one cannot use the dual approximation method of \cite{HS87,HS88}.  Our approach here is based on \cite{DBLP:journals/corr/EpsteinL14}.

The case of identical machines is the special case of uniformly related machines where all speeds are equal.  See \cite{HS87,HocBook,AAWY98} for earlier approximation schemes for this case.

The next special objective we consider here is {\em scheduling
with rejection}.  This is the special case of our objective
function where $\pi_j$ is finite (at least for some jobs).  In
\cite{BLMSS00,ES04} there is a PTAS for this variant (for $\psi\in
\{ 0,1\}$) on identical machines and on uniformly related
machines.

The last special case we consider here is the {\em machines with activation costs} model that was considered by \cite{DBLP:conf/soda/KhullerLS10}.  They considered the special case of our model with makespan objective and $\tau=2$, with $\alpha_i(1)=0$ for all $i$, and $p_j(1)=\infty$ for all $j\in J$.  In this case activating a machine as type $1$ means that the machine is not operating and cannot process any job.  For this case \cite{DBLP:conf/soda/KhullerLS10} presents a PTAS.

We summarize the previously studied special cases with reference to the previously approximation scheme of the better complexity class (i.e., we cite the first EPTAS if there is one, and the first PTAS if an EPTAS was not known prior to our work) in table \ref{table1}.

\begin{table}[h!]
  \begin{center}
{\normalsize
     \begin{tabular}{| l | c | r |}
\hline
Definition of the special case using our notation&PTAS/EPTAS& Refernce\\
\hline
\hline
$\tau=1$, $\psi=1$, and $\pi_j=\infty\ \forall j$ & EPTAS & \cite{Ja10}\\
\hline
$\tau=1$, $\psi=0$, and $\pi_j=\infty\ \forall j$ & EPTAS & \cite{EL13}\\
\hline
$\tau=1$, $\psi=1$, and $s_i=1\ \forall i$ & {\bf PTAS} &  \cite{BLMSS00}\\
\hline
$\tau=1$, $\psi=1$ or $\psi=0$ & {\bf PTAS} & \cite{ES04}\\
\hline
Machines with pre-specified type, $\psi=1$ and $\pi_j=\infty\ \forall j$ & EPTAS &   \cite{JM17v2}\\
\hline
Machines with pre-specified type, $\psi=0$, $s_i=1 \ \forall i$, and $\pi_j=\infty\ \forall j$ & {\bf PTAS} &  \cite{BW12} \\
\hline
 $\tau=2$, $\alpha_i(1)=0 \ \forall i$, $p_j(1)=\infty \ \forall j$, $\psi=1$, and  $\pi_j=\infty\ \forall j$ & {\bf PTAS } & \cite{DBLP:conf/soda/KhullerLS10}\\
\hline
 \end{tabular}
\vspace{0.2in}
   \caption{Summary of previous studies of special cases of problem $P$. For every row for which the second column is a PTAS, our EPTAS is the first efficient polynomial time approximation scheme for this special case.}
    \label{table1}
}  \end{center}
\end{table}

\paragraph{Outline of the scheme}  We apply geometric rounding of the parameters of the input (see Section \ref{sec:round}), followed by a guessing step in which we guess for each type the minimum index of the machine that is activated to this type together with its approximated load (see Section \ref{sec:guess}).  This guessing is motivated by a standard characterization of near-optimal solutions that is described earlier in Section \ref{sec:nice}.  Based on these rounding and guessing steps, we formulate a mixed integer linear program (MILP) that is solved to optimality in polynomial time using \cite{Len83,Kan83} and the property that the number of integer variables is a constant (see Section \ref{sec:MILP} for the derivation of this mathematical program), and we prove that the optimal cost to our scheduling problem $P$ is approximated by the solution obtained to the MILP.  Last, we use the solution of  the MILP to round it into a feasible solution to problem $P$ whose cost is  approximately the cost of the solution of the MILP (see Section \ref{sec:round-milp} for a description of this step and its analysis).

\section{Rounding of the input\label{sec:round}}
In what follows we would like to assume that the speed of each machine is an integer power of $1+\eps$, and that for every job $j$ and type $t$, we have that $p_j(t)$ is an integer power of $1+\eps$.  Given an instance $I$ of problem $P$ that does not satisfy these conditions, we round down the speed of each machine $i$ to an integer power of $1+\eps$, and for each job $j$ and type $t$, we round up the value of $p_j(t)$ to an integer power of $1+\eps$.  That is, we create a new rounded instance $I'$ in which the speed of machine $i$ is $s'_i$, and for each job $j$ and type $t$, we let $p'_j(t)$ be its size if it is assigned to a machine of type $t$, where we define
$$ s'_i = (1+\eps)^{\lfloor \log_{1+\eps} s_i \rfloor} \ \ \ \  \forall i \ \ \    , \ \ \ p'_j(t) = (1+\eps)^{\lceil \log_{1+\eps} p_j(t) \rceil} \  \ \ \  \ \forall j,t \ .$$
  The other parameters of the input are left in $I'$ as they were in $I$.  The analysis of this step is proved in the following lemma that follows using standard arguments.  Recall that a feasible solution to $P$ means selecting a type for each machine satisfying the total activation cost constraint and specifying a job assignment function.

\begin{lemma}\label{rounding_step_lem}
Given a feasible solution to $I$ of cost $C_I$, then the same solution is a feasible solution to $I'$ of cost (evaluated as a solution to $I'$) at most $(1+\eps)^{2\phi} \cdot C_I$.  Given a feasible solution to $I'$ of cost $C_{I'}$, then the same solution is a feasible solution to $I$ of cost (evaluated as a solution to $I$) at most $C_{I'}$.
\end{lemma}

\begin{proof}
Consider a job assignment function $\sigma$, and a selection of type $t_i$ for every machine $i$.  The feasibility conditions in $I$ and in $I'$ are the same because  the total activation budget constraint is satisfied in $I$ if and only if it is satisfied in $I'$ as the activation cost functions as well as the value of $\hat{A}$ are the same in the two instances.  It remains to consider the cost of this assignment as a solution to $I$ and as a solution to $I'$.  The total rejection penalty of the jobs rejected by $\sigma$ is the same in the two instances.

Consider machine $i$, and let $\Lambda_i$ be its load in $I$, and $\Lambda'_i$ be its load in $I'$ (with respect to the given solution).  Assume that the solution we consider activate machine $i$ as type $t$.  Then, $$\Lambda_i = \frac{ \sum_{j\in J: \sigma(j)=i} p_j(t)}{s_i} \mbox{ and } \Lambda'_i = \frac{ \sum_{j\in J: \sigma(j)=i} p'_j(t)}{s'_i} \ . $$  We have using our definition of the rounding that $s'_i \leq s_i \leq (1+\eps)s'_i$ and for all $j$, $(1+\eps) p_j(t) \geq p'_j(t) \geq p_j(t)$.

Therefore,
$$ \Lambda_i = \frac{ \sum_{j\in J: \sigma(j)=i} p_j(t)}{s_i} \leq \frac{ \sum_{j\in J: \sigma(j)=i} p'_j(t)}{s'_i} = \Lambda'_i \ \ ,$$ while $$ \Lambda'_i = \frac{ \sum_{j\in J: \sigma(j)=i} p'_j(t)}{s'_i} \leq \frac{ \sum_{j\in J: \sigma(j)=i} (1+\eps)^2 \cdot p_j(t)}{s_i}= (1+\eps)^2 \Lambda_i \ .$$

Thus, we conclude that for every machine $i$ we have $\Lambda_i \leq \Lambda'_i \leq (1+\eps)^2\Lambda_i$.  Therefore,
by Lemma \ref{claim_prop_F}
and the definition of the objective function $obj$, the claim follows.
\qed
\end{proof}

Using this lemma and noting that applying the rounding step takes linear time, we conclude that without loss of generality with a slight abuse of notation, we assume that the input instance satisfies the properties that $s_i$ and $p_j(t)$ are integer powers of $1+\eps$ (for all $i,j,t$).

\section{Characterization of near-optimal solutions\label{sec:nice}}
We say that a feasible solution to $P$ is {\em nice} (or nice solution) if the following property holds.  Let $i<i'$, be a pair of machines that are activated to a common type $t$ such that $i$ is the minimum index of a machine that is activated to type $t$, then the load of $i$ is at least the load of $i'$ times $\eps^2$.  The following lemma together with the guessing step described in the next section serve as an alternative to the dual approximation method of \cite{HS87,HS88} and suit cases in which the dual approximation method does not work (i.e., non-bottleneck load balancing problems).

\begin{lemma}\label{nice_lem}
Given an instance of $P$ and a feasible solution $\sol$ of cost $\sol$, there exists a feasible solution $\sol'$ that is a nice solution whose cost $\sol'$ satisfies $\sol'\leq (1+\eps)^{\phi}\cdot \sol$.
\end{lemma}

\begin{proof}
We apply the following process for modifying $\sol$ into $\sol'$.  The process changes the assignment of some jobs that are not rejected in $\sol$.  Thus, the value of the total rejection penalty (i.e., $\sum_{j\in J: \sigma(j)=0} \pi_j$) is left without modification.  The process is defined for every type $t$ (one type at a time) by changing the assignment of some jobs that are assigned to machines of type $t$ and are moved to the lowest index machine of type $t$.

Consider a fixed type $t$ and let $i$ be a machine of lowest index
that is assigned type $t$ in $\sol$. We will modify the set of
jobs assigned to $i$, and we let $\lambda$ denote the current load
of $i$.  We perform the following iteration until the first time
where the cost of the solution is  increasing (or we decide to stop and move to the next type).  Thus, we stop
after applying the first iteration that causes the cost of the
solution to increase. Let $i''$ be a machine of maximum load among
the machines that are assigned type $t$ in the current solution.
If $\lambda$ is at least the load of $i''$ times $\eps^2$ (and in
particular if $i''=i$), we do nothing and continue to the next
type (the condition of nice solutions is satisfied for type $t$ by
definition of maximum).  Otherwise, we move all jobs assigned to
machine $i''$ to be assigned to machine $i$.  We recalculate
$\lambda$ and the cost of the resulting solution and check if we
need to apply the iteration again (if we stop and move to the next type, then by the convexity of $F$, the new load of $i$ is larger than the load of $i''$ prior to this iteration, and by the definition of $i''$ the condition of nice solutions is satisfied for type $t$).

Consider a specific type $t$ with $i$ as defined above, and let
$i'$ be the value of $i''$ in the last iteration.  We let
$\lambda'$ be the load of $i$ just before the last iteration of
the last procedure.  Using Lemma \ref{claim_prop_F} and the
symmetry of $F$, it suffices to show that in the last iteration
the maximum of the loads of $i$ and $i'$ is increased by a
multiplicative factor of at most $1+\eps$.

There are two cases.  In the first case assume
that $s_{i'} > \eps s_i$.  Consider first moving all jobs assigned
to $i$ (prior to the last iteration) to run on $i'$.  The
resulting load of $i'$ is increased by at most
$\frac{\lambda'}{\eps}$, and this is at most $(1+\eps)$ times the
load of $i'$ prior to the last iteration.  Then by moving all the
jobs that are assigned to $i'$ to run on $i$ the load incurred by
these jobs can only decrease and the claim follows. Otherwise,
 we conclude that $s_{i'}\leq \eps s_i$. Thus, by moving
the jobs previously assigned to $i'$ to be assigned to $i$, the
total processing time of these jobs is at most $\eps$ times the
load of $i'$ (in $\sol$).  Since $i'$ is selected as the machine
of maximum load (of type $t$ in the solution prior to the last
iteration), the new load of $i$ is at most $1+\eps$ times the load
of $i'$ in the solution obtained prior to the last iteration.
\qed
\end{proof}

\section{Guessing step\label{sec:guess}}
We apply a guessing step of (partial) information on an optimal
solution (among nice solutions).  See e.g.
\cite{schuurman2001approximation} for an overview of this
technique of guessing (or partitioning the solutions space) in the
design of approximation schemes.

In what follows, we consider one nice solution of minimal cost
(among all nice solutions) to the (rounded) instance and denote
both this solution and its cost by $\opt$ together with its job
assignment function $\sigma^o$ and the type $t^o_i$ assigned (by
$\opt$) to machine $i$ (for all $i$).

The guessing is of the following information.  We guess the approximated value of the makespan in $\opt$, and denote it by $O$.  That is, if $\opt$ rejects all jobs then $O=0$, and otherwise the makespan of $\opt$ is in the interval $(O/(1+\eps),O]$.  Furthermore, for every type $t$, we guess a minimum index $\mu(t)$ of a machine of type $t$  (namely, $\mu(t)=\min_{i:t^o_i=t} i$), and its approximated load $L_t$ that is a value such that the load of machine $\mu(t)$ is in the interval $(L_t-\frac{\eps O}{\tau}, L_t]$.  Without loss of generality, we assume that $O\geq \max_t L_t$.

\begin{lemma}\label{guessing_lem}
The number of different possibilities for the guessed information on $\opt$ is $$O(nm\log_{1+\eps} n \cdot (m\tau / \eps )^{\tau}) \ . $$
\end{lemma}

\begin{proof}
To bound the number of values for $O$, let $i$ be a machine where the makespan of \opt\ is achieved and let $j$ be the job of maximum size (with respect to type $t^o_i$) assigned to $i$.  Then, the makespan of \opt\ is in the interval $[p_j(t^o_i)/s_i, n\cdot p_j(t^o_i)/s_i ]$, and thus the number of different values that we need to check for $O$ is $O(nm \log_{1+\eps} n)$.  For every type $t$ we guess the index of machine $\mu(t)$ (and there are $m$ possible such indices) and there are at most $\frac{\tau}{\eps}+1$ integer multiplies of $\eps \cdot O /\tau$ that we need to check for $L_t$ (using $L_t \leq O$).
\qed
\end{proof}
\begin{remark}
If we consider the model of machines with
pre-specified type, then we do not need to guess the value of
$\mu(t)$ (for all $t$) and the number of different possibilities
for the guessed information on $\opt$ is $O(nm\log_{1+\eps} n
\cdot (\tau / \eps )^{\tau})$.
\end{remark}

\section{The mixed integer linear program \label{sec:MILP}}
Let $\gamma \geq 10$ be a constant that is chosen later ($\gamma$ is a function of $\tau$ and $\eps$).  For a type $t$ and a real number $W$, we say that job $j$ is {\em large for $(t,W)$} if $p_j(t) \geq \eps^{\gamma} \cdot W$, and otherwise it is {\em small for $(t,W)$}.

\paragraph{Preliminaries.} Our MILP follows the configuration-MILP paradigm as one of its main ingredients.  Thus, next we define our notion of configurations.  A {\em configuration $C$} is  a vector encoding partial information regarding the assignment of jobs to one machine where $C$ consists of the following components: $t(C)$ is the type assigned to a machine with configuration $C$, $s(C)$ is the speed of a machine with configuration $C$, $w(C)$ is an approximated upper bound on the total size of jobs assigned to a machine with this configuration where we assume that $w(C)$ is an integer power of $1+\eps$ and the total size of jobs assigned to this machine is at most $(1+\eps)^3 \cdot w(C)$, $r(C)$ is an approximated upper bound on the total size of small jobs (small for $(t(C),w(C)$) assigned to a machine with this configuration where we assume that $r(C)$ is an integer multiple of $\eps \cdot w(C)$ and the total size of small jobs assigned to this machine is at most $r(C)$, last, for every integer value of $\nu$ such that $(1+\eps)^{\nu} \geq \eps^{\gamma} \cdot w(C)$ we have a component $\ell(C,\nu)$ counting the number of large jobs assigned to a machine of configuration $C$ with size $(1+\eps)^{\nu}$.  Furthermore we assume that $r(C)+ \sum_{\nu} (1+\eps)^{\nu} \cdot \ell(C,\nu) \leq (1+\eps)^3 \cdot w(C)$.  Let $\C$ be the set of all configurations.

\begin{lemma}\label{number_conf_lem}
For every pair $(s,w)$, we have $$|\{ C\in \C: s(C)=s, w(C)=w\}| \leq \tau\cdot \left( \frac{2}{\eps} \right)^{(2\gamma+1)^2 \log_{1+\eps} (1/\eps) } \ . $$
\end{lemma}
The right hand side is (at least) an exponential function of $1/\eps$ that we denote by $\beta$.

\begin{proof}
The number of types is $\tau$, so $t(C)$ has $\tau$ possible values.  The value of $r(C)$ is an integer multiple of $\eps \cdot w$ that is smaller than $(1+\eps)^3 w \leq 2w$ so there are at most $\frac 2{\eps}$ such values.  The number of different values of $\nu$ for which a job  of size $(1+\eps)^{\nu}\geq \eps^{\gamma} \cdot w$ is smaller than $2w$ is at most $\log_{1+\eps} \frac {2}{\eps^\gamma} \leq 2\gamma \cdot \log_{1+\eps} \frac {1}{\eps} $ and for each such value of $\nu$ we have that the value of $\ell(C,\nu)$ is a non-negative integer smaller than $\frac{2}{\eps^{\gamma}}$.  This proves the claim using $\gamma \geq 10$.
\qed
\end{proof}

Our MILP formulation involves several blocks and different families of variables that are presented next (these blocks have limited interaction).  We present the variables and the corresponding constraints before presenting the objective function.

\paragraph{First block - machine assignment constraints.}  For every machine $i$ and every type $t$, we have a variable $z_{i,t}$ that encodes if machine $i$ is assigned type $t$, where $z_{i,t}=1$ means that machine $i$ is assigned type $t$.  Furthermore, for every type $t$ and every speed $s$, we have a variable $m(s,t)$ denoting the number of machines of (rounded) speed $s$ that are assigned type $t$.  For every type $t$, we have $z_{i,t}=0$ for all $i< \mu(t)$ while $z_{\mu(t),t}=1$ enforcing our guessing.  The (additional) machine assignment constraints are as follows:

For every machine $i$, we require $$\sum_{t=1}^{\tau} z_{i,t} = 1 , $$ encoding the requirement that for every machine $i$, exactly one type is assigned to $i$.

For every type $t$ and speed $s$, we have $$\sum_{i:s_i = s} z_{i,t} = m(s,t) \ . $$

Let $m_s$ be the number of machines in the rounded instance of speed $s$, then $$\sum_{t=1}^{\tau} m(s,t) = m_s \ . $$

Last, we have the machine activation budget constraint
$$\sum_{t=1}^{\tau} \sum_{i=1}^m \alpha_i(t)z_{i,t} \leq \hat{A} \ .$$

The variables $z_{i,t}$ are fractional variables and their number is $O(m\tau)$, for every type $t$ and for every speed $s$ such that $s_{\mu(t)} \geq s\geq s_{\mu(t)} \cdot \eps^{\gamma}$ we require that $m(s,t)$ is an integer variable while all other variables of this family of variables are fractional.  Observe that the number of variables that belong to this family and are required to be integral (for the MILP) formulation is $O(\tau \gamma \log_{1+\eps} \frac{1}{\eps})$ that is bounded by a polynomial in $\frac {\tau \gamma}{\eps}$, and the number of fractional variables of the family $m(s,t)$ is $O(n\tau)$.

\paragraph{Second block - job assignment to machine types and rejection constraints.}
For every job $j$ and every $t\in \{ 0,1,\ldots ,\tau\}$, we have
a variable $y_{j,t}$ that encodes if job $j$ is assigned to
machine that is activated to type $t$ (for $t\geq 1$) or rejected
(for $t=0$).  That is, for $t\geq 1$, if $y_{j,t}=1$, then job $j$
is assigned to  machine of type $t$, and if $y_{j,0}=1$ it means
that $j$ is rejected (and we will pay the rejection penalty
$\pi_j$).  Furthermore for every type $t$ and every {\em possible}
integer value $\zeta$ we have two variables $n(\zeta,t)$ and
$n'(\zeta,t)$ denoting the number of jobs assigned to machine of
type $t$ whose (rounded) size (if they are assigned to machine of
type $t$) is $(1+\eps)^{\zeta}$ that are assigned as large jobs
and that are assigned as small jobs, respectively.  Here, possible
values of $\zeta$ for a given $t$ are all integers for which
$(1+\eps)^{\zeta} \leq s_{\mu(t)} \cdot \min \{ L_t/(\eps^3), O\}$
such that the rounded input contains at least one job whose size
(when assigned to a machine of type $t$) is $(1+\eps)^{\zeta}$
(where recall that $L_t$ is the guessed load of machine $\mu(t)$
and $O$ is the guessed value of the makespan).  We denote by
$\zeta(t)$ the set of possible values of $\zeta$ for the given
$t$.  We implicitly use the variables $n(\zeta,t)$ and
$n'(\zeta,t)$ for $\zeta \notin \zeta(t)$ (i.e., impossible values
of $\zeta$) by setting those variables to zero.

The constraints that we introduce for this block are as follows:

For every job $j$, we should either assign it to a machine (of one of the types) or reject it, and thus we require that $$\sum_{t=0}^{\tau} y_{j,t}=1 \ . $$
Furthermore, for every type $t$ and possible value of $\zeta$ (i.e., $\zeta \in \zeta(t)$) we require, $$\sum_{j: p_j(t)=(1+\eps)^{\zeta}} y_{j,t} \leq n(\zeta,t) +n'(\zeta,t) \ . $$

For the MILP formulation, the variables $y_{j,t}$ are fractional,
while the variables $n(\zeta,t)$ and $n'(\zeta,t)$ are  integer
variables only if $\zeta \in \zeta(t)$ and $(1+\eps)^{\zeta} \geq
s_{\mu(t)} L_t \eps^{\gamma}$ (and otherwise they are fractional).
Observe that we introduce for this block $O(n\tau)$ fractional
variables (excluding variables that are set to $0$ corresponding
to impossible values of $\zeta$) and $O(\tau \gamma \log_{1+\eps}
\frac{1}{\eps})$ integer variables.

\paragraph{Third block - configuration constraints.}
For every $C\in \C$ we have a variable $x_C$ denoting the number
of machines  of speed $s(C)$ activated to type $t(C)$ whose job
assignment is according to configuration $C$.  Furthermore,  for
every configuration $C \in \C$ and every integer value of $\nu$
such that $(1+\eps)^{\nu} < \eps^{\gamma} w(C)$ we have a variable
$\chi(C,\nu)$ denoting the number of jobs whose size (when
assigned to machine of type $t(C)$) is $(1+\eps)^{\nu}$ that are
assigned to machines of configuration $C$.  Such a variable
$\chi(C,\nu)$ exists only if there exists at least one job $j$
whose size (when assigned to a machine of type $t$) is
$(1+\eps)^{\nu}$.  For $C \in \C$, we let $\nu(C)$ denote the set
of values of $\nu$ for which the variable $\chi(C,\nu)$ exist. For
every $t$, we require that machine $\mu(t)$ has a configuration
where $s_{\mu(t)}\cdot L_t$ is approximately $w(C)$.  Thus, for
every $t$, we will have the constraint $$\sum_{C\in \C:
s(C)=s_{\mu(t)}, t(C)=t, s_{\mu(t)}\cdot L_t \leq w(C) \leq
(1+\eps)^3 \cdot s_{\mu(t)}\cdot L_t } x_C \geq 1 \ .$$

For the MILP formulation,  $x_C$ is required to be integer only if
$C$ is a {\em heavy}  configuration, where  $C$ is  {\em heavy} if
$w(C) \geq \eps^{\gamma^3} L_{t(C)} \cdot s_{\mu(t(C))}$. The
variables $\chi(C,\nu)$ are fractional for all $C\in \C$ and
$\nu\in \nu(C)$.  Observe that the number of integer variables
depends linearly in $\beta$ where the coefficient is upper bounded
by a polynomial function of $\frac{\gamma}{\eps}$.

It remains to consider the constraints bounding these variables together with the $n(\zeta,t)$, $n'(\zeta,t)$ and $m(s,t)$ introduced for the earlier blocks.  Here, the constraints have one sub-block for each type $t$.  The {\em sub-block of type $t$} (for $1\leq t \leq \tau$) consists of the following constraints:

For every type $t$ and every (rounded) speed $s$ we cannot have more than $m(s,t)$ machines with configurations satisfying $t(C)=t$ and $s(C)=s$, and therefore we have the constraint $$\sum_{C\in \C: t(C)=t, s(C) =s} x_C \leq m(s,t) \ . $$

For every $\zeta \in \zeta(t)$, we have that all the $n(\zeta,t)$ jobs of size $(1+\eps)^{\zeta}$ that we guarantee to schedule on machine of type $t$ are indeed assigned to such machine as large jobs.  Thus, we have the constraints $$\sum_{C\in \C: t(C)=t} \ell(C,\zeta) \cdot x_C = n(\zeta,t) . $$

The last constraints ensures that for every $\zeta \in \zeta(t)$, the total size of all jobs of size at least $(1+\eps)^{\zeta}$ that are scheduled as small jobs fits the total area of small jobs in configurations for which $(1+\eps)^{\zeta}$ is small with respect to $(t,w(C))$.  Here, we need to allow some additional slack, and thus for configuration $C$ we will allow to use $r(C)+2\eps w(C)$ space for small jobs.  Thus, for every integer value of $\zeta$ we have the constraint $$\sum_{\zeta' \geq \zeta} n'(\zeta',t) \cdot (1+\eps)^{\zeta'} \leq \sum_{C\in \C: t(C)=t, \eps^{\gamma} \cdot w(C) > (1+\eps)^{\zeta}} (r(C)+2\eps w(C)) x_C  \ . $$
Observe that while we define the last family of constraints to have an infinite number of constraints, we have that if when we increase $\zeta$, the summation on the left hand side is the same, then the constraint for the larger value of $\zeta$ dominates the constraint for the smaller value of $\zeta$.  Thus, it suffices to have the constraints only for $\zeta \in \cup_{t=1}^{\tau} \zeta(t)$.

In addition to the last constraints we have the non-negativity constraints (of all variables).

\paragraph{The objective function.} Using these variables and (subject to these) constraints we define the minimization (linear) objective function of the MILP as $$ \psi \cdot O  + (1-\psi) \cdot \sum_{C\in \C} \left(\frac{w(C)}{s(C)}\right)^{\phi} \cdot x_C + \sum_{j=1}^n \pi_j \cdot y_{j,0}\ .$$

Our algorithm solves optimally the MILP and as described  in the next section uses the solution for the MILP to obtain a feasible solution to problem $P$ without increasing the cost too much.  Thus, the analysis of the scheme is crucially based on the following proposition.

\begin{proposition}\label{prop-milp}
The optimal objective function value of the MILP is at most $(1+\eps)^{\phi}$ times the cost of $\opt$ as a solution to $P$.
\end{proposition}

\begin{proof}
Based on the nice solution $\opt$ to the rounded instance, specified by the type $t^o_i$ assigned to machine $i$ (for all $i$) and the job assignment function $\sigma^o$, we specify a feasible solution to the MILP as follows.  Later we will bound the cost of this feasible solution.

First, consider the variables introduced for the machine
assignment block and its constraints.  The values of $z_{i,t}$ are
as follows: $z_{i,t^o_i}=1$ and for $t\neq t^o_i$, we let
$z_{i,t}=0$.  Furthermore, for every speed $s$ and type $t$, we
let $m(s,t)$ be the number of machines of speed $s$ that are
assigned type $t$.  Then, for every type $t$, we will have $z_{i,t}=0$ for $i<\mu(t)$ and $z_{\mu(t),t}=1$ by our
guessing. For every machine $i$, we have
$\sum_{t=1}^{\tau} z_{i,t} = 1$ since every machine is assigned
exactly one type. For every type $t$ and speed $s$, we have
$\sum_{i:s_i = s} z_{i,t} = m(s,t)$, as the left hand side counts
the number of machines of speed $s$ and type $t$.  The constraint
$\sum_{t=1}^{\tau} m(s,t) = m_s$ is satisfied as every machine is
activated with exactly one type and thus contribute to exactly one
of the counters of the summation on the left hand side. Last, we
have that the machine activation budget constraint
$\sum_{t=1}^{\tau}\sum_{i=1}^m \alpha_i(t)z_{i,t} \leq \hat{A}$ is
satisfied by our assignment of values to the variables as the left
hand side is exactly the total activation cost of $\opt$, and
$\opt$ is a feasible solution to $P$.

Next, consider the other variables. For every job $j$, we let $y_{j,0}=1$ if $j$ is rejected by $\opt$, and otherwise we let
$y_{j,t}=1$ if and only if the following holds for some value of
$i$ $$\sigma^o(j)=i \ \ \mbox{and} \ \ t^o_i=t \ .$$ Observe that
the total rejection penalty of the jobs that are rejected by
$\opt$ is exactly $\sum_j \pi_j y_{j,0}$. Next, we assign a
configuration to every machine based on $\opt$.  Consider a
specific machine $i$, we let $C(i)$ be the configuration we define
next.  $t(C(i))=t^o_i$ is the type assigned to machine $i$ by
$\opt$ and $s(C(i))=s_i$ is its speed. The value of $w(C(i))$ is
computed by rounding up the total size of jobs assigned to $i$ in
$\opt$ to the next integer power of $1+\eps$, $r(C(i))$ is
computed by first computing the total size of small jobs assigned
to $i$ and then rounding up to the next integer multiple of $\eps
\cdot w(C(i))$, last, for every integer value of $\nu$ such that
$(1+\eps)^{\nu} \geq \eps^{\gamma} \cdot w(C(i))$, the component
$\ell(C(i),\nu)$ counts the number of jobs assigned to $i$ whose
size is $(1+\eps)^{\nu}$.  Thus, the requirement $r(C(i))+
\sum_{\nu} (1+\eps)^{\nu} \cdot \ell(C(i),\nu) \leq (1+\eps)^3
\cdot w(C(i))$ is satisfied as the left hand side exceeds the
total size of jobs assigned to $i$ in $\opt$ by at most $\eps
w(C(i))$ and the right hand side exceeds the total size of jobs
assigned to $i$ in $\opt$ by at least $3\eps w(C(i))$. For every
configuration $C\in \C$ we let $x_C$ be the number of machines
whose assigned configuration is $C$. Furthermore, for every $C\in
\C$ and every $\nu \in \nu(C)$, we calculate the number of jobs of
size $(1+\eps)^{\nu}$ (when assigned to machines of type $t(C)$)
that are assigned by $\opt$ to machines whose assigned
configuration is $C$, and we let $\chi(C,\nu)$ be this number.  By
our guessing, we conclude that $$\sum_{C\in \C: s(C)=s_{\mu(t)},
t(C)=t, s_{\mu(t)}\cdot L_t \leq w(C) \leq (1+\eps)^3 \cdot
s_{\mu(t)}\cdot L_t } x_C \geq 1$$ is satisfied for every type
$t$. For every type $t$ and every $\zeta \in \zeta(t)$, we let
$n(\zeta,t)$ be $n(\zeta,t)=\sum_{i: t(C(i))=t} \ell(C(i),\zeta)$,
and $n'(\zeta,t)$ be defined as $n'(\zeta,t) = \max \{ 0, \sum_{j:
p_j(t)=(1+\eps)^{\zeta}} y_{j,t} - n(\zeta,t)\}$.  This completes
the assignment of values to the variables that we consider.

Then, since for every job $j$, if $\sigma^o(j) \geq 1$, then
$\sigma^o(j)$ has exactly one type, and otherwise $y_{j,0}=1$, by
definition of the values of the $y$-variables, we have
$\sum_{t=0}^{\tau} y_{j,t}=1$. Furthermore, for every type $t$ and
$\zeta \in \zeta(t)$), by definition of $n'(\zeta,t)$, we have
$$\sum_{j: p_j(t)=(1+\eps)^{\zeta}} y_{j,t} \leq n(\zeta,t)
+n'(\zeta,t) \ . $$

For every type $t$ and every (rounded) speed $s$, $\opt$ does not
have more than $m(s,t)$ machines with configurations satisfying
$t(C)=t$ and $s(C)=s$, and therefore we have the constraint
$\sum_{C\in \C: t(C)=t, s(C) =s} x_C \leq m(s,t)$. For every $t$
and every $\zeta \in \zeta(t)$, we have that $$\sum_{C\in \C:
t(C)=t} \ell(C,\zeta) \cdot x_C = \sum_{i:t(C(i))=t}
\ell(C(i),\zeta) = n(\zeta,t) \ , $$ where the first equality
holds by changing the order of summation (using the definition of
$x_C$) and the second holds by the definition of the value of
$n(\zeta,t)$. Last, for every machine $i$ and every $\zeta \in
\zeta(t(C(i)))$, the total size of all jobs of size at least
$(1+\eps)^{\zeta}$ that are scheduled as small jobs on machine $i$
is smaller than $r(C(i))$.  Thus, for every integer value of
$\zeta$, we have  $$\sum_{\zeta' \geq \zeta} n'(\zeta',t) \cdot
(1+\eps)^{\zeta'} \leq \sum_{C\in \C: t(C)=t, \eps^{\gamma} \cdot
w(C) > (1+\eps)^{\zeta}} (r(C)+2\eps w(C)) x_C \ . $$ We summarize
that the solution we found is a feasible solution to the MILP
(where all variables are integer, so the integrality constraints
of the MILP hold as well).

Last, consider the objective function value of the solution we
constructed. It is $ \psi \cdot O  + (1-\psi) \cdot \sum_{C\in \C}
\left(\frac{w(C)}{s(C)}\right)^{\phi} \cdot x_C + \sum_{j=1}^n
\pi_j \cdot y_{j,0} $. By our guessing, we conclude that $O$ is at
most $1+\eps$ times the makespan of $\opt$.  By definition, for
machine $i$ whose assigned configuration is $C(i)$ and its load in
$\opt$ is $\Lambda^o_i$ we have $\frac{w(C(i))}{s_i} \leq
(1+\eps)\Lambda^o_i$. Therefore, we have that
 \begin{eqnarray*} && \psi \cdot O  + (1-\psi) \cdot \sum_{C\in \C} \left(\frac{w(C)}{s(C)}\right)^{\phi} \cdot x_C + \sum_{j=1}^n \pi_j \cdot y_{j,0}\\   &\leq& (1+\eps)^{\phi} \cdot F(\Lambda^o_1, \ldots ,\Lambda^o_m) + \sum_{j=1}^n \pi_j \cdot y_{j,0} \ , \end{eqnarray*} and using the fact that the total rejection penalty in $\opt$ equals $\sum_{j=1}^n \pi_j \cdot y_{j,0}$, the claim follows.
\qed
\end{proof}

\section{Transforming the solution to the MILP into a schedule \label{sec:round-milp}}
Consider the optimal solution $(z^*,m^*, y^*,n^*,n'^*,x^*,\chi^*)$
for the MILP, our first step is to round up each component of
$n^*$ and $n'^*$.  That is, we let $\hat{n}(\zeta,t)=\lceil
n^*(\zeta,t) \rceil$ and $\hat{n}'(\zeta,t) =\lceil n'^*(\zeta,t)
\rceil$ for every $\zeta$ and every $t$.

Furthermore, we solve the following linear program (denoted as $(LP-y)$) that has totally unimodular constraint matrix and integer right hand side, and let $\hat{y}$ be an optimal integer solution for this linear program:
\begin{eqnarray*}
\min &\sum_{j=1}^n \pi_j \cdot y_{j,0} & \\
\mbox{subject to}& \sum_{t=0}^{\tau} y_{j,t}=1 & \forall j\in J,\\ &\sum_{j\in J: p_j(t)=(1+\eps)^{\zeta}} y_{j,t} \leq \hat{n}(\zeta,t) +\hat{n}'(\zeta,t)& \forall t \in\{ 1,2,\ldots ,\tau\}  \ , \  \forall \zeta\in \zeta(t) \} \ , \\
&y_{j,t}\geq 0 & \forall j\in J \ , \ \forall t\in \{ 0,1,\ldots ,\tau\} \ .
\end{eqnarray*}
We will assign jobs to types (and reject some of the jobs) based on the values of $\hat{y}$, that is if $\hat{y}_{j,t}=1$ we will assign $j$ to a machine of type $t$ (if $t\geq 1$) or reject it (if $t=0$). Since $y^*$ is a feasible solution to $(LP-y)$ of cost that equal the total rejection penalty of the solution to the MILP, we conclude that the total rejection penalty of this (integral) assignment of jobs to types is at most the total rejection penalty of the solution to the MILP.  In what follows we will assign $\hat{n}(\zeta,t)+\hat{n}'(\zeta,t)$ jobs of size $(1+\eps)^{\zeta}$ to machines of type $t$ (for all $t$).

The next step is to round up each component of $x^*$, that is, let
$\hat{x}_C=\lceil x^*_C \rceil$, and allocate $\hat{x}_C$ machines
of speed $s(C)$ that are activated as type $t(C)$ and whose
schedule {\em follows} configuration $C$.  These $\hat{x}_C$
machines are partitioned into $x'_C=\lfloor x^*_C \rfloor$ {\em
actual machines} and $\hat{x}_C-x'_C$ {\em virtual machines}. Both
actual and virtual machines are not machines of the instance but
{\em temporary machines} that we will use for the next step.

\begin{lemma}\label{job-assignment-lemma}
It is possible to construct (in polynomial time) an allocation of
$\hat{n}(\zeta,t)$ jobs of size $(1+\eps)^{\zeta}$ for all
$t,\zeta$ to (actual or virtual)  machines that follow
configurations in $\{ C\in \C : t(C)=t, (1+\eps)^{\zeta} \geq
\eps^{\gamma} \cdot w(C)\ \}$, and of $\hat{n}'(\zeta,t)$ jobs of
size  $(1+\eps)^{\zeta}$ for all $t,\zeta$ to (actual or virtual)
machines that follow configurations in $\{ C\in \C :
t(C)=t,(1+\eps)^{\zeta} < \eps^{\gamma} \cdot w(C)\ \}$, such that
for every machine that follows configuration $C\in \C$, the
total size of jobs assigned to that machine is at most $(1+\eps)^7
w(C)$.
\end{lemma}

\begin{proof}
First, we allocate the large jobs.  That is, for every value of $t$ and $\zeta$ we allocate $\hat{n}(\zeta,t)$ jobs whose size (when assigned to machine of type $t$) is $(1+\eps)^{\zeta}$ to (actual or virtual) machines that follow configurations in $\{ C\in \C : t(C)=t, (1+\eps)^{\zeta} \geq \eps^{\gamma} \cdot w(C)\ \}$.  We allocate for every machine that follow configuration $C$ exactly $\ell(C,\zeta)$ such jobs (or less if there are no additional jobs of this size to allocate).  Since $\sum_{C\in \C: t(C)=t} \ell(C,\zeta) \cdot x^*_C = n^*(\zeta,t)$, we allocate in this way at least  $n^*(\zeta,t)$ such jobs as large jobs.  By the constraint $\sum_{C\in \C: t(C)=t} \ell(C,\zeta) \cdot x_C = n(\zeta,t)$, we have the following
 \begin{eqnarray*}
  \sum_{C\in \C: t(C)=t} \ell(C,\zeta) \cdot \hat{x}_C &\geq& \sum_{C\in \C: t(C)=t} \ell(C,\zeta) \cdot x^*_C\\
  &=& n^*(\zeta,t) \\
  &>& \hat{n}(\zeta,t) -1
  \end{eqnarray*}
Since both sides of the last sequence of inequalities are integer
number, we conclude that  $\sum_{C\in \C: t(C)=t} \ell(C,\zeta)
\cdot \hat{x}_C \geq \hat{n}(\zeta,t)$, and thus all these
$\hat{n}(\zeta,t)$ jobs  are assigned to machines that follow
configurations in $\{ C\in \C : t(C)=t, (1+\eps)^{\zeta} \geq
\eps^{\gamma} \cdot w(C)\ \}$ without exceeding the bound of
$\ell(C,\zeta)$ on the number of jobs of size $ (1+\eps)^{\zeta}$
assigned to each such machine.

Next, consider the allocation of small jobs.  We apply the
following process for each type separately. We first allocate one
small job of each size to machine  $\mu(t)$ where here we mean
that if machine $\mu(t)$ follows configuration $C(\mu(t))$, then
we will allocate one job of each size of the form
$(1+\eps)^{\zeta}$ for all $\zeta$ for which $(1+\eps)^{\zeta-3} <
\eps^{\gamma} w(C(\mu(t)))$.  Observe that these values of $\zeta$
are the only values for which we may have $\hat{n}'(\zeta,t)\neq
n'^*(\zeta,t)$. Let $\tilde{n}'(\zeta,t)$ denote the number of
jobs of size $(1+\eps)^{\zeta}$ that we still need to allocate.
Since $\gamma \geq 10$, this step increases the total size of jobs
assigned to machine $\mu(t)$ by $\eps^{\gamma} \cdot w(C(\mu(t)))
\cdot \sum_{h=0}^{\infty} \frac{1}{(1+\eps)^{h-3}} \leq \eps
w(C(\mu(t)))$.

Given a type $t$, we sort the machines that follow configurations with type $t$ according to the values of $w(C)$ of the configuration $C$ that they follow.  We sort the machines in a monotonically non-increasing order of $w(C)$.  Similarly, we sort the collection of $\tilde{n}'(\zeta,t)$ jobs of size  $(1+\eps)^{\zeta}$ (for all values of $\zeta$) in a non-decreasing order of $\zeta$.  We allocate the jobs to (actual or virtual) machines using the next fit heuristic.  That is, we start with the first machine (in the order we described) as the current machine.  We pack one job at a time to the current machine, whenever the total size of the jobs that are assigned to the current machine (that follows configuration $C$) exceeds $r(C)+2\eps w(C)$ (it is at most $r(C)+3\eps w(C)$ as we argue below), we move to the next machine and define it as the current machine.  We need to show that all jobs are indeed assigned in this way and that whenever we pack a job into the current machine that follows configuration $C$, the size of the job is smaller than $\eps^{\gamma} w(C)$.  If we append one (non-existing) extra machine of type $t$ that follows a configuration with $w(C) < \min_{j} p_{j}(t) / \eps^{\gamma}$, then it suffices to show that whenever we move to a new current machine  that follows configuration $C$, all the (small) jobs of size $(1+\eps)^{\zeta}$ for $\zeta\in \{ \zeta': (1+\eps)^{\zeta'} \geq \eps^{\gamma} w(C)\}$ are assigned.  This last required property holds, as whenever the last set of values of $\zeta$ is changed, we know that the total size of the (small) jobs we already assigned is (unless all these jobs are assigned) at least
\begin{eqnarray*}
&& \sum_{C\in \C: t(C)=t, \eps^{\gamma} \cdot w(C) > (1+\eps)^{\zeta}} (r(C)+2\eps w(C)) \hat{x}_C \\
 &\geq& \sum_{C\in \C: t(C)=t, \eps^{\gamma} \cdot w(C) > (1+\eps)^{\zeta}} (r(C)+2\eps w(C)) x^*_C \\
 &\geq& \sum_{\zeta' \geq \zeta} n'^*(\zeta',t) \cdot (1+\eps)^{\zeta'} \\
 &\geq& \sum_{\zeta' \geq \zeta} \tilde{n}'(\zeta',t) \cdot (1+\eps)^{\zeta'} \ .
   \end{eqnarray*}
   By allocating a total size of at most $r(C)+3\eps w(C)$ of small jobs to a machine that follows configuration $C$, the resulting total size of jobs assigned to that machine is at most $$(1+\eps)^3 w(C)+ 3\eps w(C) \leq (1+\eps)^6 w(C)$$ and the claim follows.
\qed
\end{proof}

The assignment of jobs for which $y_{j,0} \neq 0$ to machines is specified by assigning every job that was assigned to a virtual machine that follows configuration $C$ to machine $\mu(t(C))$ instead, and assigning the jobs allocated to actual machines by allocating every actual machine to an index in $\{ 1,2,\ldots ,m\}$ following the procedure described in the next step.  Before describing the assignment of actual machines to indices in $\{ 1,2,\ldots ,m\}$ of machines in the instance (of problem $P$), we analyze the increase of the load of machine $\mu(t)$ due to the assignment of jobs that were assigned to virtual machines that follow configuration with type $t$.

\begin{lemma}\label{virtual_machine_alloc_lem}
There is a value of $\gamma$ for which the resulting total size of
jobs assigned to machine $\mu(t)$ is at most $(1+\eps)^8 \cdot
w(C(\mu(t))$ where machine $\mu(t)$ follows the configuration
$C(\mu(t))$.
\end{lemma}

\begin{proof}
Since $x^*_C$ is forced to be integral for all heavy
configurations, we conclude that if there exists a virtual machine
that follows configuration $C$ with type $t(C)=t$, then $w(C)\leq
\eps^{\gamma^3} \cdot L_{t(C)}\cdot s_{\mu(t(C))} \leq
2\eps^{\gamma^3} w(C(\mu(t)))$, where the last inequality holds
using the constraint $\sum_{C\in \C: s(C)=s_{\mu(t)}, t(C)=t,
s_{\mu(t)}\cdot L_t \leq w(C) \leq (1+\eps)^3 \cdot
s_{\mu(t)}\cdot L_t } x_C \geq 1$ and allocating configuration
$C(\mu(t))$ that causes $x^*$ to satisfy this inequality,  to machine
$\mu(t)$. For each configuration $C$, there is at most one virtual
machine that follows $C$, and since there are at most $\beta$
configurations with a common component of $w(C)$ and the given
type $t$ and speed $s$, we conclude using the fact that $\opt$ is
nice that the total size of jobs that we move from their virtual
machines to machine $\mu(t)$ is at most
$$2\beta \cdot \eps^{\gamma^3} w(C(\mu(t))) \cdot
\sum_{h=0}^{\infty}  \frac{1}{(1+\eps)^h} \cdot
\sum_{h'=0}^{\infty} \frac{1}{\eps^2}\frac{1}{(1+\eps)^{h'}} \leq
\beta \cdot \eps^{\gamma^3-5} w(C(\mu(t))) \ .$$  The claim will
follow if we can select a value of $\gamma$ such that $\beta<
\eps^{6-\gamma^3}$.

Recall that $\beta = \tau\cdot \left( \frac{2}{\eps} \right)^{(2\gamma+1)^2 \log_{1+\eps} (1/\eps) }\leq \tau\cdot \left( \frac{2}{\eps} \right)^{5\gamma^2 /\eps^2} $.  Thus, in order to ensure that $\beta< \eps^{6-\gamma^3}$, it suffices to select $\gamma$ such that $\tau < \left( \frac{\eps}{2} \right)^{(5\gamma^2 /\eps^2) + 6-\gamma^3}$.  Observe that for every $\eps>0$ the function $H(\gamma)= (5\gamma^2 /\eps^2) + 6- \gamma^3$ decreases without bounds when $\gamma$ increases to $\infty$.  Thus, we can select a value of $\gamma$ (as a function of $\tau$ and $\eps$) for which the last inequality holds, e.g. selecting $\gamma = \tau \cdot \frac{20}{\eps^2}$ is sufficient.
\qed
\end{proof}

Next, we describe the assignment of actual machine to indices in $\{1,2,\ldots ,m\}$.  More precisely, the last step is to assign a type $\hat{t}_i$ for every machine $i$ satisfying the total activation cost bound, and to allocate for every $C\in \C$ and for every actual machine that follows configuration $C$, an index $i$ such that $\hat{t}_i=t(C)$.
This assignment of types will enforce our guessing of $\mu(t)$ for all $t$
This step is possible as we show next  using the integrality of the assignment
polytope.
\begin{lemma}\label{machine_assignment_types_feas_lem}
There is a polynomial time algorithm that finds a type $t_i$ for
every machine $i$, such that the total activation cost of all
machines is at most $\hat{A}$, for all $s,t$ the number of machines of
speed $s$ that are activated to type $t$ is at least the number of
actual machines that follow configurations with type $t$ and speed
$s$, and for all $t$  $\mu(t)$ is the minimum index of a machine that is assigned type $t$.
\end{lemma}
\begin{proof}
We consider the following linear program (denoted as $(LP-z)$):
\begin{eqnarray*}
\min & \sum_{t=1}^{\tau} \sum_{i=1}^m \alpha_i(t)z_{i,t} & \\
s.t. & \sum_{t=1}^{\tau} z_{i,t}=1 & \forall i, \\
 & \sum_{i: s_i=s} z_{i,t} \geq \sum_{C\in \C : t(C)=t, s(C)=s} x'_C& \forall s,\ \forall t\\
 & z_{i,t} =0 & \forall t, \forall i<\mu(t) \\
 & z_{\mu(t),t}=1 & \forall t \\
 & z_{i,t} \geq 0 & \forall i \forall t \ \ .
\end{eqnarray*}
The constraint matrix of $(LP-z)$ is totally unimodular, and thus
by solving the linear program and finding an optimal basic
solution, we find an optimal integer solution in polynomial time.
Since the fractional solution $z^*$ is a feasible solution with
objective function value that does not exceed $\hat{A}$, we
conclude that the optimal integer solution that we find does not
violate the upper bound on the total activation cost.  This
integer solution defines a type $t_i$ for every machine $i$ by
letting $t_i$ be the value for which $z_{i,t_i}=1$.  Then using
the constraints of the linear program for every speed $s$ and
every type $t$, the number of machines of speed $s$ for which we
define type $t$ is at least the number of actual machines that
follow configurations with type $t$ and speed $s$ and furthermore for every type $t$ 
$\mu(t)$ is the minimum index of a machine that is assigned type $t$, as required. \qed
\end{proof}

Thus, we conclude:
\begin{theorem}\label{main-thm}
Problem $P$ admits an EPTAS.
\end{theorem}

\begin{proof}
The time complexity of the scheme is
$O(f(\frac{1}{\eps},\gamma,\tau) \cdot m^{O(\tau)} \cdot
poly(n))$, and as proved in Lemma \ref{virtual_machine_alloc_lem},
$\gamma$ is a function of $\eps$ and $\tau$.  Thus, in order to
show that the algorithm is an EPTAS, it suffices to prove its
approximation ratio, and that the resulting solution is feasible.
Based on the sequence of lemmas, the approximation ratio is
proved, using the fact that the load of an empty set of jobs is
zero no matter what is the type of the machine and thus for every
$i$, if machine $i$ is assigned an empty set of jobs the objective
function value does not depend on the type assigned to $i$.  Thus,
in the last step the cost of the solution does not increase.

Thus, we need to show that the resulting solution is feasible.  Note that every job assignment is feasible, the feasibility of the assignment of types to machines is feasible using Lemma \ref{machine_assignment_types_feas_lem}.
  Thus, our solution is a feasible solution to problem $P$.
\qed
\end{proof}

\bibliographystyle{abbrv}

\end{document}